\documentclass{scrartcl}

\usepackage[top=2.5cm, bottom=3.5cm, left=2.5cm, right=2.5cm]{geometry}

\usepackage[hidelinks, citecolor=blue]{hyperref}
\hypersetup{colorlinks=true, urlcolor=blue, linkcolor=black}

\usepackage[numbers,sort&compress]{natbib}
\usepackage{amssymb, amsthm, mathtools, braket}

\newtheorem{theorem}{Theorem}
\newtheorem{lemma}[theorem]{Lemma}
\numberwithin{equation}{section}
\numberwithin{theorem}{section}

\newcommand{\bm}{{\mathbf{m}}}
\newcommand{\bh}{{\mathbf{h}}}
\newcommand{\bJ}{{\mathbf{J}}}
\newcommand{\bw}{{\mathbf{w}}}
\newcommand{\bs}{{\mathbf{s}}}
\newcommand{\rr}{{\mathbb{R}}}
\newcommand{\cc}{{\mathbb{C}}}
\newcommand{\zz}{{\mathbb{Z}}}
\newcommand{\pp}{{\mathbb{P}}}
\newcommand{\ee}{{\mathbb{E} \,}}
\newcommand{\oh}{{\mathcal{O}}}
\newcommand{\im}{{\operatorname{Im}\,}}
\newcommand{\re}{{\operatorname{Re }\,}}
\newcommand{\mean}[1]{{{\langle #1 \rangle}}}
\newcommand{\parder}[2]{{\frac{\partial #1}{\partial #2}}}
\newcommand{\beq}[1]{\begin{equation} \label{#1}}
\newcommand{\eeq}{\end{equation}}

\begin{document}
\addtokomafont{author}{\raggedright}
\title{\raggedright On the mean-field equations for ferromagnetic spin systems\vspace{-.3cm}}
\author{\hspace{-.2cm}Christian Brennecke and Per von Soosten}
\date{\vspace{-.7cm}}
\maketitle

\minisec{Abstract} We derive mean-field equations for a general class of ferromagnetic spin systems with an explicit error bound in finite volumes. The proof is based on a link between the mean-field equation and the free convolution formalism of random matrix theory, which we exploit in terms of a dynamical method. We present three sample applications of our results to Ka\'{c} interactions, randomly diluted models, and models with an asymptotically vanishing external field.
\bigskip

\section{Introduction}
The subject of this note is a ferromagnetic spin system with Hamiltonian $H\vcentcolon \{-1, 1\}^N \to \rr$ defined by
\beq{eq:firsthdef} H(\sigma) =  -\frac{1}{2}\sum_{ij} J_{ij} \sigma_i \sigma_j - \sum_i h_i \sigma_i, \qquad J_{ij} \geq 0 \eeq
and Gibbs expectation of $f\vcentcolon  \{-1, 1\}^N \to \rr$ given by
\[\mean{f} = \frac{1}{Z} \sum_\sigma f(\sigma) e^{-H(\sigma)}, \qquad Z =  \sum_\sigma e^{-H(\sigma)}. \]
There is no loss of generality in assuming that $J_{ij} = J_{ji}$, which we will do from now on. The results below will be meaningful in the setting where $\max J_{ij} \to 0$ as $N \to \infty$.

The prototypical example of such a system is the Curie-Weiss model, which corresponds to the choice $J_{ij} = \beta N^{-1}$ with $\beta \geq 0$ and $h_i = h \in \rr$ (see~\cite{MR2189669} for a comprehensive discussion and bibliography). The phase diagram of the Curie-Weiss model can be obtained by studying the observable $M = N^{-1} \sum_i \sigma_i$, whose Gibbs expectation satisfies the mean-field equation
\beq{eq:cwmf} \mean{M} = \tanh(h + \beta \mean{M})\eeq
in the thermodynamic limit $N \to \infty$. To understand~\eqref{eq:cwmf} heuristically, we note that
\[\mean{\sigma_i} = \Braket{\tanh\left(h + \frac{\beta}{N} \sum_{j \neq i} \sigma_j \right) }_i\]
where $\mean{\cdot}_i$ is the Gibbs expectation with the spin $\sigma_i$ removed. Since one expects $M$ to concentrate around its expectation both at high temperature (spins are approximately independent) and at low temperature when $h \neq 0$ (almost all spins take on the same value), it should be possible to pull the Gibbs expectation inside the $\tanh$. Although this intuition should extend also to the more general case~\eqref{eq:firsthdef}, rigorous treatments of~\eqref{eq:cwmf} rely strongly on the symmetries of the Curie-Weiss model. The classical approach (see~\cite[Ch.\,2]{MR3752129}) is to compute the entropy using that $H = -\frac{\beta}{2} N M^2 - h N M$ essentially depends on only one degree of freedom. Other proofs exploit this symmetry using the Hubbard-Stratonovich transformation~\cite{PhysRevLett.3.77,1957SPhD....2..416S}, Varadhan's lemma~\cite{MR203230}, or the exchangeability of certain spin configurations~\cite{MR2707160, MR2288072}.

We propose a simple dynamical method for establishing explicit finite-volume versions of the mean-field equations that are valid for general interactions. Setting $m_i = \mean{\sigma_i}$, the analogue of~\eqref{eq:cwmf} in the general setting is
\beq{eq:genmf} \bm = \tanh(\bh + \bJ \bm) \eeq
where $\bm = (m_i)$, $\bh = (h_i)$, $\bJ = (J_{ij})$, and the $\tanh$ of a vector is defined in an entrywise sense. To state our main result, we also introduce the notation
\[\hat{\bh} = \min h_i, \qquad \|\bJ\|_{\infty, \infty} = \max_j \sum_i J_{ij}, \qquad \|\bJ\|_{1, \infty} = \max_{ij} J_{ij}.\]
The typical mean-field setting corresponds to $\|\bJ\|_{\infty, \infty} = \oh(1)$ and $\|\bJ\|_{1, \infty} = \oh(N^{-1})$ as $N \to \infty$.

\begin{theorem} \label{thm:main} Let $\bh \geq 0$. Then
\[ \|\bm - \tanh(\bh + \bJ \bm)\|_\infty  \le \frac{\|\bJ\|_{1, \infty}}{\hat{\bh}} \left(3 + \|\bJ\|_{\infty, \infty} + \log\left(1 + \frac{\|\bJ\|_{\infty, \infty}}{\hat{\bh}} \right)  \right).\]
\end{theorem}

Theorem~\ref{thm:main} is only informative when $\hat{\bh} > 0$. However, when $\|\bJ\|_{1,\infty} \to 0$ as $N \to \infty$, the theorem is sufficient to study the physically relevant order of limits that first lets $N \to \infty$ and then lets $\bh \to 0$. Moreover, one can arrange for both $\hat{\bh}$ and the right hand side in Theorem~\ref{thm:main} to  also vanish asymptotically if $\|\bJ\|_{1, \infty}$ does -- a point we will return to in Section~\ref{sec:applications}. To the best of our knowledge, quantitative bounds like Theorem~\ref{thm:main} are relatively scarce in the standard literature on the Curie-Weiss model. Stein's method for exchangeable pairs~\cite{MR2707160, MR2288072} yields concentration bounds for $M$ under the Gibbs measure with fluctuations of order $N^{-1/2}$. In particular, these bounds prove Theorem~\ref{thm:main} for the Curie-Weiss model with an error rate of $N^{-1/2}$. When $\bh > 0$, our result improves the error rate to $N^{-1}$ while remaining valid for models with general interactions $\bJ$. There have also been recent developments concerning the fluctuations of general nonlinear functions of Bernoulli random variables whose gradients are close to a low-dimensional manifold \cite{MR3519474}. The application of these ideas to mean-field Gibbs measures was further explored in the works~\cite{MR3663625, pmlr-v75-jain18b, fluct-preprint}.

The general mean-field equations~\eqref{eq:genmf} are very common in the physics literature~\cite[Sec.\,3.2]{MR1007980}. Nevertheless, it was noted in~\cite[Sec.\,3.4]{MR2707160} that completely general interactions $\bJ$ seem to pose significant challenges for the existing mathematical strategies. It was also shown in~\cite[Thm.\,3.5]{MR2707160} that methods based on exchangeable pairs can yield analogues of the mean-field equations for certain conditional averages with high probability under the Gibbs measure, from which~\eqref{eq:genmf} follows at sufficiently high temperatures. Exchangeable pairs have also been used to analyze the rank-one case $\bJ = \mathbf{w}\mathbf{w}^\intercal$ with a regular and non-negative $\mathbf{w} \in \rr^N$~\cite{MR4046512}. A different perspective can be found in the works~\cite{MR1989669,MR2219531}, which prove that the magnetizations of sufficiently high-dimensional or long-range systems approximately minimize the free energy of the associated mean-field theory. These are very strong results that, among many other things, imply the approximate validity of the mean-field equations but also rely on the powerful input of infrared bounds derived from reflection positivity.

There is a strong analogy between the mean-field equation and the subordination relations in random matrix theory and free probability. The central example of the latter is concerned with the resolvent $G(t, z) = (A_t - z)^{-1}$ of an $N \times N$ matrix  $A_t  = A_0 + \sqrt{t} \Phi$, where $A_0$  is a diagonal matrix and $\Phi$ is drawn from the Gaussian orthogonal ensemble. The main assertion, that the limiting empirical eigenvalue distribution of $A_t$ is given by the free convolution of the limiting empirical eigenvalue distribution  of $A_0$ and a semicircular element, can be captured by the fact that
\[s_0(z) = \frac{1}{N} \sum_{\lambda \in \sigma(A_0)} \frac{1}{\lambda-z}, \qquad s_t(z) = \frac{1}{N} \sum_{\lambda \in \sigma(A_t)} \frac{1}{\lambda - z}\]
are analytic functions that map the half-plane $\im z > 0$ into itself and satisfy
\beq{eq:subordination} s_t(z) = s_0(z + t s_t(z))\eeq
in the $N \to \infty$ limit~\cite{MR0475502}. To illustrate the connection to the mean-field equation, we consider the simplest case of~\eqref{eq:genmf} when $h_i = h > 0$ and $\sum_j J_{ij} = \beta$ for all $i$. In this case, one can construct a solution of~\eqref{eq:genmf} by letting each entry solve the scalar equation
\beq{eq:tracemf} m = \tanh(h + \beta m). \eeq
It is not hard to show that the positive solution $m = m(h)$ of~\eqref{eq:tracemf} extends to an analytic function of $h$ that maps the half-plane $\re h > 0$ into itself. It follows that $\tilde{m}(z) = i m(-iz)$ is an analytic function of $\{\im z >  0\}$ into itself with
\[\tilde{m}(z) = \tan (z + \beta \tilde{m}(z)).\]
It is a consequence of the Herglotz trick~\cite{MR3823190} that
\[\tan(z) =  \sum_{\lambda \in \Lambda} \frac{1}{\lambda - z}, \qquad \Lambda =  \pi (\zz + 1/2),\]
so $\tilde{m}$ is the Stieltjes transform of the free convolution of $\sum_{\lambda \in \Lambda} \delta_\lambda$ and a semicircular element. Expanding on this theme, the assertion of Theorem~\ref{thm:main} that the relation~\eqref{eq:tracemf} is not only valid asymptotically but that the error is small even in finite volumes when $\hat{\bh} \gg \|\bJ\|_{1, \infty}$, is analogous to the local deformed semicircle law~\cite{MR3134604}. The proof of the local law depends crucially on the fact that
\beq{eq:wardid} \left| \parder{}{z} G_{ii}(t, z) \right| \le \sum_{k} |G_{ik}(t, z)|^2 =  \frac{\im G_{ii}(t, z)}{\im z}.\eeq
Lemma~\ref{thm:corrbound}, which rests on the Lee-Yang theorem~\cite{PhysRev.87.410}, contains a similar inequality for $m_i$ that lies at the heart of Theorem~\ref{thm:main}.

It was observed in~\cite{MR0475502} that the subordination relation~\eqref{eq:subordination} is equivalent to the partial differential equation
\beq{eq:pde} \parder{}{t} s_t(z) = s_t(z) \parder{}{z} s_t(z).\eeq
In finite volumes, $s_t$ exactly satisfies a perturbed version of this transport equation, which enables a simple proof the local law by combining the analytic structure contained in~\eqref{eq:wardid} with an approximate characteristic curve~\cite{MR3920502,MR4049087}. Our approach to Theorem~\ref{thm:main} is to generalize this analysis by exploiting the identity
\beq{eq:corederid} \parder{}{J_{ij}} \mean{f} = m_i \parder{}{h_j} \mean{f} + m_j \parder{}{h_i} \mean{f} + \parder{^2}{h_i \partial h_j} \mean{f}. \eeq
The relationship between~\eqref{eq:corederid} and~\eqref{eq:pde} is most apparent in the Curie-Weiss model, where~\eqref{eq:corederid} yields
\[\parder{}{\beta} \mean{M} = \mean{M} \parder{}{h} \mean{M} + \frac{1}{2N} \parder{^2}{h^2} \mean{M},\]
which is the evolution studied in~\cite{MR3920502} without the stochastic terms. We note that similar differential \textit{inequalities} have been studied in a variety of models that go beyond the mean-field setting~\cite{MR894398,MR874906}.

This paper is structured as follows. In Section~\ref{sec:proof}, we use~\eqref{eq:corederid} to derive a transport equation in the general setting and prove the analogue of~\eqref{eq:wardid}. This allows us to extend the ideas of~\cite{MR3920502,MR4049087} concerning approximate characteristic curves to the present setting and prove Theorem~\ref{thm:main}. Then, in Section~\ref{sec:applications}, we present three sample applications of our method to Ka\'{c} interactions, randomly diluted models, and models with an asymptotically vanishing external field.

\section{Proof of Theorem~\ref{thm:main}}\label{sec:proof}
Rearranging indices shows that Theorem~\ref{thm:main} follows if we can prove that
\beq{eq:toprove}|m_1 - \tanh(h_1 + \bJ_1^\intercal \bm)| \le \frac{\|\bJ_1\|_\infty}{\hat{\bh}} \left(3 + \|\bJ_1\|_1 + \log\left(1 + \frac{\|\bJ_1\|_1}{\hat{\bh}} \right)  \right) \eeq
where $\bJ_1 = (J_{1i})$ is the first column of $\bJ$. By symmetry, we can write $H(\sigma) = H(1, \sigma)$, where
\[H(t, \sigma) = t \sum_i J_{1i} \sigma_1 \sigma_i + \sum_i h_i \sigma_i + H_1(\sigma_2, \dots, \sigma_N)\]
and $H_1$ does not depend on $\sigma_1$. Moreover, if $\bm(t, \bh)$ denotes the vector of magnetizations under $H(t, \cdot)$, the identity~\eqref{eq:corederid} yields
\beq{eq:realpde} \parder{}{t} \bm(t, \bh) = m_1(t, \bh)  \parder{}{\bJ_1}  \bm(t, \bh) + \parder{^2}{h_1 \partial \bJ_1}   \bm(t, \bh)  + (\bJ_1^\intercal \bm) \parder{}{h_1} \bm(t, \bh),\eeq
where $\parder{}{\bJ_1} = \bJ_1^\intercal \nabla_{\bh}$ is the directional derivative with respect to $\bh$ in direction $\bJ_1$.

If we knew that the first two terms on the right hand side of~\eqref{eq:realpde} were negligible, then~\eqref{eq:realpde} would reduce to a transport equation that could be solved by varying $h_1$ along a suitable characteristic curve. We will prove the appropriate bounds for this purpose with the help of the following integral representation. Suppose that $f$ is a holomorphic function defined on a half-plane $\re z > -\kappa$ such that $\re f \geq 0$ and such that $f$ remains bounded as $z \to \infty$ along $\rr$. Then,
\[f(z) = a + \int_{\rr} \! \frac{1}{z + \kappa - i\lambda } \, \mu(d\lambda)\]
for some positive measure $\mu$ and $a \in \cc$~\cite[Ch.\,5]{MR1307384}, which implies that
\[\left|f^\prime(z) \right| \le \frac{\re f(z)}{\re z + \kappa}.\]
We note that the use of integral representations of holomorphic functions to bound correlations is a classical idea (see for instance~\cite{MR2905800, MR432092}).
\begin{lemma} \label{thm:corrbound} For every $k$ and every $\bs \geq 0$, 
\[ \left| \parder{}{\bs}  m_k(t, \bh) \right| \le \frac{\|\bs\|_\infty}{\hat{\bh}} \, m_k(t, \bh)\]
and
\[\left| \parder{^2}{h_1 \partial \bs}  m_k(t, \bh) \right| \le \frac{\|\bs\|_\infty}{\hat{\bh}} \,\frac{m_k(t, \bh)}{h_1}  \]
for  all $t \geq 0$ and $\bh > 0$.
\end{lemma}
\begin{proof} If we fix $h_j$ with $\re h_j > 0$ for $j \neq k$, the meromorphic function $f(h_k) = m_k(t, \bh)$ satisfies $f^\prime(h_k) = 1 - f^2(h_k)$ wherever it is analytic. Hence, combining the Picard-Lindel\"{o}f theorem with analytic continuation shows that there is some $\Gamma \in \cc$ such that
\[f(h_k) = \tanh(h_k + \Gamma)\]
for all $h_k \in \cc$. If $\re \bh > 0$, the Lee-Yang theorem asserts that the partition function $Z(t, \bh) = \sum_{\sigma} e^{-H(t,\sigma)}$ cannot vanish and therefore $f$ cannot have a pole in this region. This is only possible when $\re \Gamma \geq 0$ and therefore
\[\re m_k(t, \bh) = \re f(h_k) \geq 0\]
whenever $\re \bh \geq 0$.

To prove the first bound, we fix $\bh \geq 0$ and consider the function
\[f(z) = m_k(t, \bh + z\bs).\]
Then $\re f > 0$ on the half-plane $\re z >  -\kappa$ with $\kappa = \hat{\bh} / \|\bs\|_\infty$ so
\[\left| \parder{}{\bs}  m_k(t, \bh) \right| = \left| f^\prime(0) \right| \le \frac{\re f(0)}{\kappa} = \frac{\re m_k(t, \bh)}{\kappa}.\]
Since $m_k(t, \bh)$ is real when $\bh$ is real, this is the first assertion of the lemma.

For the second bound, we use the auxiliary function
\[g(z) = \partial_1 m_k(t, \bh + z\bs).\]
Then $g$ is also holomorphic on $\re z > -\kappa$ and considering the right hand side as a function of $h_1$ shows that
\[|g(z)| \le \frac{\re m_k(t, \bh + z\bs)}{\re h_1}.\]
Letting $C$ be a positively oriented circle about $0$ of radius $\kappa$, we have
\[  \parder{^2}{h_1 \partial \bs} m_k(t, \bh) = g^\prime(0) = \oint_C \! \frac{g(\xi)}{\xi^2} \, d\xi\]
so
\[ \left|\parder{^2}{h_1 \partial \bs} m_k(t, \bh) \right| \le \frac{1}{2\pi \kappa} \int_0^{2\pi} \frac{\re m_k(\bh + e^{i\theta} \bs)}{\re h_1} \, d\theta = \frac{\re m_k(t,\bh)}{\kappa \, \re h_1} \]
using the mean value property of the harmonic function $z \to \re m_k(t, \bh + z\bs)$.
\end{proof}

With Lemma~\ref{thm:corrbound} in place, we now fix $\bh \geq 0$ and define an approximate characteristic curve $\bw = (w_i)$ for~\eqref{eq:realpde} by
\[\parder{}{t} w_i(t) = \begin{cases}-\bJ_1^\intercal \bm(t, \bw(t)) & i = 1\\ 0 & \mbox{else} \end{cases}, \qquad \bw(1) = \bh. \]
Since $\bm$ is non-negative and uniformly Lipschitz continuous in $\bh \geq 0$, such a curve exists and satisfies $\hat{\bw}(t) \geq \hat{\bh}$ for all $t \in [0, 1]$. The following lemma shows that any weighted average $\bs^\intercal \bm$ does not significantly change along the curve $\bw(t)$, provided that $\|\bs\|_\infty$ is small.

\begin{lemma}\label{thm:avgbound} Let $\bs \geq 0$. Then
\[\sup_{t \in [0,1]} \left| \bs^\intercal \bm(\bh) - \bs^\intercal \bm(t, \bw(t)) \right| \le \frac{\|\bs\|_\infty}{\hat{\bh}} \left(\|\bJ_1\|_1 + \log\left(1 + \frac{\|\bJ_1\|_1}{\hat{\bh}} \right) \right).\]
\end{lemma}
\begin{proof}
Inserting the characteristic curve into~\eqref{eq:realpde} and multiplying by $\bs$, we obtain
\beq{eq:jfluct} \bs^\intercal \bm(\bh) - \bs^\intercal \bm(t, \bw(t)) = \int_t^1 \! m_1(r, \bw(r))  \parder{}{\bJ_1}  \bs^\intercal \bm(r, \bw(r)) + \parder{^2}{h_1 \partial \bJ_1}   \bs^\intercal \bm(r, \bw(r)) \, dr.\eeq
Combining the identity
\[\parder{}{\bJ_1} \bs^\intercal \bm = \parder{}{\bs} \bJ_1^\intercal \bm\]
with Lemma~\ref{thm:corrbound} shows that the first integral on the right hand side of~\eqref{eq:jfluct} is bounded by
\[\int_t^1 \! \left| m_1(r, \bw(r))  \parder{}{\bJ_1}  \bs^\intercal \bm(r, \bw(r)) \right| \, dr \le \frac{\|\bs\|_\infty \|\bJ_1\|_1}{\hat{\bh}},\]
whereas the second integral is bounded by
\beq{eq:secondint}\int_t^1 \! \left|  \parder{^2}{h_1 \partial \bJ_1}   \bs^\intercal \bm(r, \bw(r)) \right| \, dr \le \frac{\|\bs\|_\infty}{\hat{\bh}} \int_t^1 \! \frac{\bJ_1^\intercal \bm(r, \bw(r))}{w_1(r)} \, dr.\eeq
The right hand side of~\eqref{eq:secondint} can be calculated explicitly since $\parder{}{t} w_1(t) = - \bJ_1^\intercal \bm(t, \bw(t))$, which yields a final bound of
\[ \int_t^1 \! \left|  \parder{^2}{h_1 \partial \bJ_1}   \bs^\intercal \bm(r, \bw(r)) \right| \, dr \le \frac{\|\bs\|_\infty}{\hat{\bh}} \log\left( \frac{w_1(t)}{h_1} \right)\le \frac{\|\bs\|_\infty}{\hat{\bh}} \log\left(1 + \frac{\|\bJ_1\|_1}{\hat{\bh}} \right) .\]
\end{proof}

The evolution of $m_1$ along the characteristic curve is given by
\[\parder{}{t} m_1(t, \bw(t)) =  m_1(t, \bw(t))  \parder{}{\bJ_1}  m_1(t, \bw(t)) + \parder{^2}{h_1 \partial \bJ_1}   m_1(t, \bw(t)).\]
By Lemma~\ref{thm:corrbound},
\[\left| m_1(t, \bw(t))  \parder{}{\bJ_1}  m_1(t, \bw(t)) \right| \le \frac{\|\bJ_1\|_\infty}{\hat{\bh}}\]
and
\[\left| \parder{^2}{h_1 \partial \bJ_1}   m_1(t, \bw(t)) \right| = 2 \left| m_1(t, \bw(t)) \parder{}{\bJ_1}  m_1(t, \bw(t))   \right| \le \frac{2\|\bJ_1\|_\infty}{\hat{\bh}},\]
so
\beq{eq:outerbound}\left| m_1(\bh) - \tanh(w_1(0)) \right| = \left| m_1(1, \bw(1)) - m_1(0, \bw(0)) \right| \le \frac{3\|\bJ_1\|_\infty}{\hat{\bh}}. \eeq
Since
\[w_1(0) = h_1 + \int_0^1 \! \bJ_1^\intercal \bm(t, \bw(t)) \, dt,\]
Lemma~\ref{thm:avgbound} with $\bs = \bJ_1$ implies that
\[\left| w_1(0) - h_1 - \bJ_1^\intercal \bm(\bh) \right| \le \frac{\|\bJ_1\|_\infty}{\hat{\bh}} \left( \|\bJ_1\|_1 + \log\left(1 + \frac{\|\bJ_1\|_1}{\hat{\bh}} \right) \right).\]
Inserting this into~\eqref{eq:outerbound} and using the Lipschitz continuity of $\tanh$ completes the proof of~\eqref{eq:toprove}.

\section{Three applications}\label{sec:applications}
Our first two applications consist of showing that two classes of models have the same thermodynamic behavior as the Curie-Weiss model. In the first, we consider a system in a box $\Lambda \subset \zz^d$ with Ka\'{c} interactions
\[J_{ij} = \beta \lambda^d f(\lambda(i-j)), \qquad h_i = h > 0, \qquad i, j \in \Lambda\]
where $f$ is a bounded Riemann integrable probability density on $\rr^d$ and $\beta, \lambda > 0$. Ka\'{c} interactions are  ``physical'' in the sense that they yield a convex free energy, but this free energy still converges to the convex envelope of the Curie-Weiss free energy as $\lambda \to 0$~\cite{MR187835, MR148416}. This fact can be used to provide a justification of Maxwell's equal-area rule for the van der Waals isotherm (see also \cite[Ch.\,4]{MR3752129} for further details). There is also an extensive literature on Ka\'{c} interactions with fixed $\lambda > 0$ and related models -- we refer the reader to~\cite{MR1694123, MR1453742, MR1935654, MR1414119, MR2460018} and references therein.

Writing $\bm_\Lambda$ for the vector of magnetizations corresponding to the box $\Lambda$, Theorem~\ref{thm:main} asserts that
\[\|\bm_\Lambda - \tanh(h + \bJ \bm_\Lambda)\|_{\infty} \le C \frac{\lambda^d}{h \log h} \]
for some absolute constant $C < \infty$. Translation invariance and standard convexity arguments show that there is some $m \in \rr$ such that for any fixed $i \in \zz^d$ we have $m_{\Lambda, i} \to m$  as $\Lambda \to \zz^d$. By the dominated convergence theorem the limit $m$ still satisfies
\[\left| m - \tanh\left(h + \beta m \sum_{i} \lambda^d f(\lambda i) \right)\right| \le C \frac{\lambda^d}{h \log h}\]
and therefore $m = \tanh(h + \beta m)$ in the $\lambda \to 0$ limit.

The second model we consider is the randomly diluted model where 
\[J_{ij} = \frac{\beta}{Np} \epsilon_{ij}, \qquad h_i = h > 0\]
and $\epsilon_{ij}$ are independent (up to symmetry) Bernoulli random variables with $\ee \epsilon_{ij} = p$. In the case where $p = p(N)$ is chosen such that $Np \to \infty$ as $N \to \infty$, it was shown in~\cite{MR1239568} that the limiting magnetizations coincide with those of the Curie-Weiss model. Under this assumption, the variance of a weighted average is bounded by
\beq{eq:bernoullivar} \ee \left| \bJ_1^\intercal \mathbf{x} - \frac{\beta}{N} \sum_i x_i \right|^2  \le \frac{\beta^2 \|\mathbf{x}\|_\infty^2}{Np} \to 0.\eeq
Applying the single-site bound~\eqref{eq:toprove}, it follows that
\[m_1 - \tanh(h + \bJ_1^\intercal \bm) \to 0\]
in $L^2(\pp)$ as $N \to \infty$. We write $m = N^{-1} \sum_i m_i$ and let $m_i^{(1)}$ denote the Gibbs mean of $\sigma_i$ with $\sigma_1$ removed. Using Lemma~\ref{thm:avgbound} and repeating the bound above yields
\[\bJ_1^\intercal \bm(h) - \sum_{i > 1} J_{1i} m_i^{(1)}(h) \to 0\]
and
\[m - \frac{1}{N} \sum_{i > 1} m_i^{(1)}(h) \to 0\]
in $L^2(\pp)$. Since $m_i^{(1)}$ is independent of $\bJ_1$, combining this with the bound~\eqref{eq:bernoullivar} shows that $\bJ_1^\intercal \bm - \beta m \to 0$ in $L^2(\pp)$. We conclude that both
\[m_1 - \tanh(h + \beta m) \to 0, \qquad m - \tanh(h + \beta m) \to 0\]
in $L^2(\pp)$.

Finally, we consider a generic model with an asymptotically vanishing external field
\[\|\bJ\|_{1, \infty} = o(1), \qquad \hat{\bh} = \|\bJ\|_{1, \infty}^{\frac{1}{2} - \delta}.\]
We assume a low-temperature condition of the form that there exists $\alpha > 1$, independent of $N$, such that for all $\mathbf{x} \geq 0$ there is some index $i$ with
\beq{eq:lowtemp}(\bJ \mathbf{x})_i \geq \alpha \mathbf{x}_i.\eeq
Provided that $\|\bJ\|_{\infty, \infty} = \oh(1)$, Theorem~\ref{thm:main} implies that
\[\tanh(\bh + \bJ \bm) = \bm + \boldsymbol \epsilon, \qquad \| \boldsymbol \epsilon\|_\infty = \oh\left(-\|\bJ\|_{1, \infty}^{\frac{1}{2} + \delta} \log \|\bJ\|_{1, \infty} \right). \]
There exists some $K > 0$ such that $\alpha \tanh(x) > x$ when $x \in (0, K)$. If it were true that $\bm \le K$, combining this with the fact that $\bh, \boldsymbol \epsilon \to 0$ and $\bh \gg \boldsymbol \epsilon$ as $N \to \infty$ would imply that
\[\bJ \bm < \alpha \bm\]
for sufficiently large $N$, contradicting~\eqref{eq:lowtemp}. We conclude that an external field strength of $\hat{\bh} = \|\bJ\|_{1, \infty}^{\frac{1}{2} - \delta}$ is sufficient to select a positive Gibbs state in the sense that
\beq{eq:posgibbs} \liminf_{N \to \infty} \max_i m_i > 0.\eeq
Degenerate examples like the case where $J_{1i} = 0$ for all $i$ demonstrate that, in general, one cannot hope for a stronger statement than~\eqref{eq:posgibbs}. For the Curie-Weiss model, the previous argument shows that an external field strength of $h = N^{-\frac{1}{2} + \delta}$ is sufficient to select the positive Gibbs state below the critical temperature of $\beta = 1$. The book~\cite[Sec.\,III.1]{MR1026102} mentions that actually an external field strength of $h = N^{-1 + \delta}$ already suffices, but we have not been able to locate a mathematical proof of this assertion in the literature.

\bigskip
\minisec{Acknowledgments}
We thank M.\,Biskup, G.\,Genovese, and S.\,Warzel for their helpful comments. The work of P.\,S. is supported by the DFG grant SO 1724/1-1.

\bibliographystyle{abbrv}
\bibliography{References}
\bigskip
\bigskip
\begin{minipage}{0.5\linewidth}
\noindent Christian Brennecke\\
Department of Mathematics\\
Harvard University\\
\verb+brennecke@math.harvard.edu+ 
\end{minipage}%
\begin{minipage}{0.5\linewidth}
\noindent Per von Soosten\\
Department of Mathematics\\
Harvard University\\
\verb+vonsoosten@math.harvard.edu+ 
\end{minipage}
\end{document}